\documentclass[11pt,letterpaper]{article}

\usepackage[backend=biber,
style=alphabetic,
sorting=nyt,
minalphanames=3,
maxbibnames=99]{biblatex}
\usepackage{subcaption}
\usepackage{amsthm,thmtools,amsmath,amsfonts,amssymb}
\usepackage{hyperref}

\usepackage{graphicx} %
\usepackage{xcolor}
\usepackage{amsfonts}
\usepackage{bm}
\usepackage{pifont}
\usepackage{dsfont}
\usepackage[linesnumbered]{algorithm2e}
\usepackage{graphicx}
\usepackage{framed}
\usepackage{csquotes}
\usepackage{enumitem}
\usepackage[margin=1in]{geometry}
\newtheorem{theorem}{Theorem}[section]
\newtheorem{corollary}[theorem]{Corollary}

\newtheorem{lemma}[theorem]{Lemma}
\newtheorem{definition}[theorem]{Definition}

\usepackage{xcolor}
\usepackage[draft]{fixme}
\fxusetheme{color}
\fxsetface{inline}{\small}
\usepackage[capitalize]{cleveref}

\DeclareMathOperator{\E}{E}

\newcommand{\local}{LOCAL\xspace}
\newcommand{\qlocal}{quantum-\local}
\newcommand{\detlocal}{det-\local}
\newcommand{\randlocal}{rand-\local} 
\newcommand{\slocal}{S\local}

\newcommand{\nonsign}{non-signaling\xspace}

\newcommand{\expect}[1]{\mathbb{E}\left[#1\right]}

\DeclareMathOperator{\dist}{dist} %
\newcommand{\neighborhood}{\NN}
\newcommand{\neigh}{\neighborhood}
\DeclareMathOperator{\poly}{poly} %

\newcommand{\view}{\VV}

\renewcommand{\AA}{\mathcal{A}}

\newcommand{\CC}{\mathcal{C}}

\newcommand{\EE}{\mathcal{E}}
\newcommand{\FF}{\mathcal{F}}

\newcommand{\HH}{\mathcal{H}}

\newcommand{\LL}{\mathcal{L}}

\newcommand{\NN}{\mathcal{N}}

\newcommand{\PP}{\mathcal{P}}

\newcommand{\VV}{\mathcal{V}}

\newcommand{\sthat}{\;|\;}

\DeclareMathOperator{\inpt}{in}
\DeclareMathOperator{\oupt}{out}
\newcommand{\localVar}{\mathrm{x}}
\newcommand{\problem}{\Pi}

\newcommand{\outcome}{\mathrm{O}}
\newcommand{\lpproblem}{\PP}

\makeatletter
\DeclareMathOperator{\regularlog}{log}
\renewcommand{\log}{\protect\@ifstar{\regularlog^*}{\regularlog}}
\makeatother

\newcommand{\Vin}{\VV_{\text{in}}}
\newcommand{\Ein}{\EE_{\text{in}}}
\newcommand{\Vout}{\VV_{\text{out}}}
\newcommand{\Eout}{\EE_{\text{out}}}

\newcommand{\outLbl}{\ell_\text{out}}

\newcommand{\lift}{\mathrm{lift}}
\newcommand{\lbadgraph}{\mathsf {badGraph}}
\newcommand{\lpromise}{\mathsf {promise}}
\newcommand{\linearizable}{\mathsf {linearizable}}

\newcommand{\lM}{\mathsf{M}}
\newcommand{\lB}{\mathsf{B}}
\newcommand{\lA}{\mathsf{A}}
\newcommand{\lP}{\mathsf{P}}
 
\begin{document}

\begin{flushleft}
	\huge\bf
	New Limits on Distributed Quantum Advantage: Dequantizing Linear Programs
\end{flushleft}
\smallskip
\begin{flushleft}
	\setlength{\parskip}{3pt}
	
	\textbf{Alkida Balliu} · Gran Sasso Science Institute
	
	\textbf{Corinna Coupette} · Aalto University
	
	\textbf{Antonio Cruciani} · Aalto University
	
	\textbf{Francesco d'Amore} · Gran Sasso Science Institute
	
	\textbf{Massimo Equi} · Aalto University
	
	\textbf{Henrik Lievonen} · Aalto University
	
	\textbf{Augusto Modanese} · Aalto University
	
	\textbf{Dennis Olivetti} · Gran Sasso Science Institute
	
	\textbf{Jukka Suomela} · Aalto University
\end{flushleft}
\smallskip
\paragraph{Abstract.}
In this work, we give two results that put new limits on distributed quantum advantage in the context of the LOCAL model of distributed computing:
\begin{enumerate}
	\item We show that there is no distributed quantum advantage for any linear program. Put otherwise, if there is a quantum-LOCAL algorithm $\AA$ that finds an $\alpha$-approximation of some linear optimization problem $\Pi$ in $T$ communication rounds, we can construct a classical, deterministic LOCAL algorithm $\AA'$ that finds an $\alpha$-approximation of $\Pi$ in $T$ rounds. As a corollary, all classical lower bounds for linear programs, including the KMW bound, hold verbatim in quantum-LOCAL.
	\item Using the above result, we show that there exists a locally checkable labeling problem (LCL) for which quantum-LOCAL is strictly weaker than the classical deterministic SLOCAL model.
\end{enumerate}
Our results extend from quantum-LOCAL to finitely dependent and non-signaling distributions, and one of the corollaries of our work is that the non-signaling model and the SLOCAL model are incomparable in the context of LCL problems: By prior work, there exists an LCL problem for which SLOCAL is strictly weaker than the non-signaling model, and our work provides a separation in the opposite direction. 
	
\thispagestyle{empty}
\setcounter{page}{0}
\newpage
\section{Introduction}

In this work, we explore the landscape of distributed graph algorithms in two dimensions:
\begin{enumerate}
	\item Classical distributed algorithms vs.\ distributed \emph{quantum} algorithms.
	\item Combinatorial graph problems (e.g., maximum independent set) vs.\ their \emph{fractional} linear-programming relaxations (e.g., maximum \emph{fractional} independent set).
\end{enumerate}
We prove two results that put limits on distributed quantum advantage; see \cref{fig:overview} for a schematic overview:
\begin{enumerate}
	\item We show that there is no distributed quantum advantage for any linear program.
	\item Using the above result, we give a new separation between quantum algorithms and classical algorithms (more precisely, between quantum-LOCAL and SLOCAL models).
\end{enumerate}

\begin{figure}
	\centering
	\includegraphics[page=1,scale = 0.8]{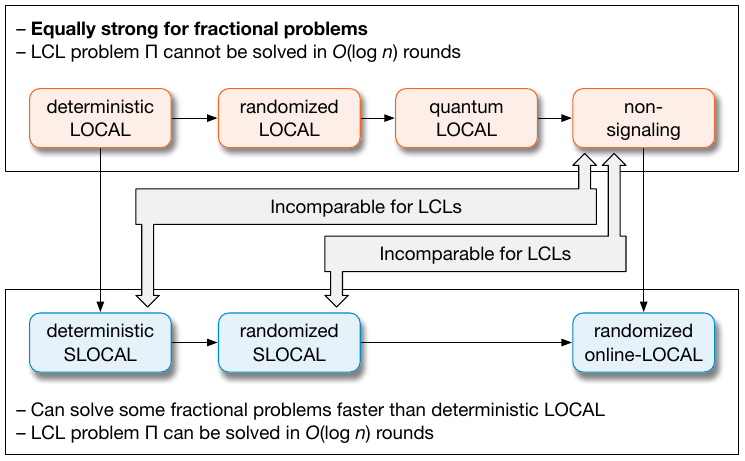}
	\caption{Overview of the results and relevant models of computing.}\label{fig:overview}
\end{figure}

\subsection{Contribution 1: Dequantizing Fractional Algorithms}\label{ssec:intro-deq}

\subparagraph*{Setting: Fractional Problems in the LOCAL Model and the Quantum-LOCAL Model.}

Let us first recall what fractional linear-programming relaxations of graph problems are. For example, in the maximum independent set problem, the task is to label each node $v \in V$ with a value $x_v \in \{0,1\}$ such that for each edge $\{u,v\}$, we satisfy $x_u + x_v \le 1$, and we are maximizing $\sum_v x_v$. Now the maximum \emph{fractional} independent set problem is the obvious linear-programming relaxation, where the range of values is $x_v \in [0,1]$. See \cref{ssec:distr-lp} for more details.

In the LOCAL model of distributed computing, 
a set of computers (nodes) communicates via bidirectional links (edges) defined by an input graph, and computation proceeds in synchronous rounds. 
In each round, each computer may send a message of unlimited size to each of its neighbors and update its state based on the messages it receives, 
and the main complexity measure is the number of communication rounds required to solve the given problem.
The quantum-LOCAL model is like the LOCAL model, except that we replace all computers with quantum computers and all communication links with quantum communication links capable of exchanging qubits.
The main source of potential advantage of the quantum-LOCAL over LOCAL comes from the fact that the messages may contain qubits entangled with the state of the computer, 
and this has indeed been used to show an advantage for LCL problems~\cite{balliu2025quantum-lcl}.

\subparagraph*{New Result.}

We prove the following result that is applicable to any fractional linear-programming relaxation:
\begin{oframed}
	\noindent
	Assume $\AA$ is a distributed algorithm that finds an $\alpha$-approximation of some linear optimization problem $\Pi$ in $T$ communication rounds in the \textbf{quantum-LOCAL model}. Then there is also an algorithm $\AA'$ that finds an $\alpha$-approximation of $\Pi$ in $T$ communication rounds in the \textbf{classical deterministic LOCAL model}.
\end{oframed}
\noindent
It was already well-known that distributed graph algorithms for solving linear programs can be \emph{derandomized} for free---without losing anything in the running time or approximation ratio. What we show is that they can also be \emph{dequantized}. This can be pushed even further, beyond the quantum-LOCAL model: We show that it holds even if $\AA$ is an algorithm in the \emph{non-signaling model}, which is a model strictly stronger than quantum-LOCAL (see \cref{sec:models} for detailed definitions).

\subparagraph*{Technical Overview.}
The proof is near-trivial: $\AA'$ outputs the \emph{expected value} of the output of $\AA$. A bit more precisely, let $X_v$ be the random variable that represents the output of $\AA$ at node $v$, and let $x_v = \expect{X_v}$ be the output produced by $\AA'$. Now if we look at the entire output vector, it holds that $x = \expect{X}$. In particular, $x$ is a linear combination of $\alpha$-approximate solutions $X$ to problem $\Pi$, and hence $x$ is also an $\alpha$-approximate solution to~$\Pi$. We give the details in \cref{sec:dequant}.


\subparagraph*{Implications for the KMW Bound.}

One of the seminal lower bounds for distributed graph algorithms is the KMW
lower bound by Kuhn, Moscibroda, and Wattenhofer \cite{KMW}; see Appendix~\ref{apx:kmw}\@. In 2009, Gavoille,
Kosowski, and Markiewicz \cite{gavoille2009} observed that this lower bound
holds also for quantum-LOCAL and non-signaling models. However, to our knowledge,
the proof was never published---they merely write that \enquote{by careful analysis, it
	is easy to~prove}.

However, now we no longer need to do the careful analysis. The key observation is that the KMW bound is inherently a bound for fractional problems. As we now know that there is no quantum advantage for fractional problems, we know that all the implications of the KMW bound indeed hold verbatim also in the quantum-LOCAL model.

While this is not a new result, at least now there is a proof explicitly written down, and the proof is fundamentally different from the idea of carefully inspecting the inner workings of the KMW bound and checking that the argument holds.

\subsection{Contribution 2: A New Separation for SLOCAL vs.\ Non-Signaling Distributions}

\subparagraph*{Setting: LCL Problems and SLOCAL Model.}

Let us now move on from fractional optimization problems to \emph{locally checkable labeling problems} or LCLs. First defined by Naor and Stockmeyer in the 1990s \cite{NaorS95}, LCLs constitute a particularly simple family of graph problems that manages to capture many problems of interest in our field. LCL problems are graph problems that can be described by listing a \emph{finite} set of valid labelings in local neighborhoods. There is a long line of work on understanding the landscape of LCL problems, e.g., \cite{NaorS95,chang_exponential_2019,dahal2023,chang_time_2019,rozhon_polylogarithmic-time_2020,brandt_automatic_2019,akbari23online-local,akbari25online-quantum,coiteux-roy24}, and this line of research has recently started to explore the interplay between various models of distributed computing, including not only the classical LOCAL model, quantum-LOCAL model, and non-signaling distributions, but also models such as SLOCAL and online-LOCAL.

The SLOCAL model \cite{ghaffari2017} is a sequential counterpart of the LOCAL model: An adversary queries nodes in a sequential order, and when a node $v$ is queried, the algorithm has to choose the final label of node $v$. To do that, the algorithm can gather the full information on its radius-$T$ neighborhood, store all this information at node $v$ (for the benefit of other nodes nearby that get queried later), and use all this information to choose its label. We refer to \cref{sec:preliminaries} for precise definitions, and to \cref{fig:overview} for an overview of how SLOCAL is related to other recently-studied models of computing.

\subparagraph*{New Result.}

The SLOCAL model is at least as strong as the LOCAL model because with the full knowledge of the radius-$T$ neighborhood, it can simulate any LOCAL algorithm for $T$ steps. However, its exact relation with the quantum-LOCAL model and non-signaling distributions has been an open question---in essence, the question is whether the ability to manipulate qubits is more useful or less useful than the ability to process nodes in some sequential order.
We prove the following new result:
\begin{oframed}
	\noindent
	There exists an LCL problem $\Pi$ such that $\Pi$ can be solved with locality $O(\log n)$ in the SLOCAL model, but it \emph{cannot} be solved with locality $O(\log^{1.49} n)$ in the non-signaling model or the quantum-LOCAL model.
\end{oframed}

\subparagraph*{Technical Overview.}

While this may seem disconnected from the results in \cref{ssec:intro-deq}, it is, in a sense, a direct corollary. One implication of the KMW bound is that the problem of finding a maximal matching in a bipartite graph of degree at most $\Delta$ cannot be solved in $o(\log \Delta / \log \log \Delta)$ rounds. On the other hand, this is a problem that is trivial to solve in the SLOCAL model in $O(1)$ rounds. So at this point, we have a \emph{family} of LCL problems $\Pi'(\Delta)$ parameterized by~$\Delta$ that is trivial in SLOCAL but nontrivial in the non-signaling model. We can now plug this family of problems into the construction of \cite{balliu2025quantum-lcl}, as maximal matchings satisfy their key technical requirement of \emph{linearizability}. Deploying this machinery yields a single genuine LCL problem that is strictly easier to solve in SLOCAL than in the non-signaling and quantum-LOCAL models.  We give the details in \cref{sec:separation}.

\subparagraph*{Implications for the Landscape of Models.}

By recent work \cite{balliu25shared-rand}, there is an LCL problem that is strictly easier to solve in the non-signaling model than in the SLOCAL model. Here, we have obtained a separation in the converse direction. Therefore, in particular:
\begin{oframed}
	\noindent
	The SLOCAL model and the non-signaling model are incomparable in the context of LCL problems; neither is able to simulate the other with constant overhead. 
\end{oframed}
\noindent
We contrast this with, e.g., the situation for the randomized online-LOCAL model, i.e., the SLOCAL model augmented with a global memory that all nodes can access~\cite{akbari23online-local}.
The randomized online-LOCAL model \emph{can} simulate both the SLOCAL model and the non-signaling model \cite{akbari25online-quantum}.


\section{Preliminaries}\label{sec:preliminaries}
Natural numbers are denoted as $\mathbb{N}$ and include $0$, while we use $\mathbb{N}_+ = \mathbb{N}\setminus\{0\}$ for natural numbers without $0$. Moreover, we use the notation $[n]=\{1,\dots,n\}$ for any $n\in\mathbb{N}_+$.
Throughout this work, we consider only undirected graphs. A graph \(G=(V,E)\) consists of a set of nodes \(V\) and a set of edges \(E\), and we use the notation \(V(G)\) and \(E(G)\), respectively, if we need to specify which graph we refer to.
Given a subset of nodes \(A\), \(G[A]\) is the subgraph induced by \(A\), that is, \(G[A]=(A,E_A)\), where \((u,v)\in E_A\) if and only if \(u,v\in A\) and \((u,v)\in E\). 

For any two nodes \(u,v \in V\) in graph \(G=(V,E)\), \(\dist_G(u,v)\) is the length of a shortest path starting from \(u\) and ending at \(v\). If there is no ambiguity, we write \(\dist(u,v)\) without the graph subscript. For subsets of nodes \(A,B \subseteq V\), we can then define \(\dist(A,B) = \min_{(u,v)\in A \times B} \dist(u,v)\). For \(T\in\mathbb{N}\), we define the radius-\(T\) neighborhood of a node \(u\) as \(\neighborhood_T[u]=\{v \sthat \dist(u,v) \le T\}\), i.e., the set of nodes at distance at most \(T\) from \(u\). This can be extended to a subset of nodes \(A \subseteq V\) as \(\neighborhood_T[A]= \cup_{u \in A} \neighborhood_T[u]\).
All the above notions of distances and neighborhoods can be easily generalized to the case where we consider the edges instead of the nodes of a graph.

If \(v\) is a node and \(e=\{u,v\}\) is an edge, then the pair \((v,e)\) is a \emph{half-edge}, and \(\HH(G)\) is the set of all pairs \((v,e)\) where \(v \in V(G)\) and \((v,e)\) is a half-edge. 
A \emph{centered graph} is a pair \((G,c)\), where \(G\) is a graph and \(c \in V(G)\) is one of its nodes, called the \emph{center}. The \emph{eccentricity} of \((G,c)\) is the maximum distance \(\max_{u \in V(G)} \dist_G(c,u)\) of a node from the center.

We now introduce the notion of \emph{labeled graph}.
\begin{definition}[Labeled graph \cite{balliu2025quantum-lcl}]\label{def:labeled-graph}
	Suppose that \(\VV\) and \(\EE\) are sets of labels.
	A graph \(G = (V,E)\) is said to be \((\VV,\EE)\)-\emph{labeled} if the following holds:
	\begin{enumerate}
		\item Each node \(v \in V\) is assigned a label from \(\VV\);
		\item Each half-edge \((v,e) \in V \times E\) satisfying $v \in e$ is assigned a label from \(\EE\).
	\end{enumerate}
\end{definition}

We will consider problems that, in general, take as input a labeled graph and require as output a labeling such that some constraints are satisfied. 
To formalize the class of problems we consider, for a \((\VV,\EE)\)-labeled graph \(G\), let \(\ell(v) \in \VV\) be the label assigned to node \(v\) and \(\ell((v,e))\in \EE\) be the label assigned to half-edge \((v,e)\).

To be able to consider labelings restricted to a subset of nodes, we say that, if we are given sets \(A\) and \(B\), a subset \(A' \subseteq A\), and a function \(f : A \to B\), then \(f \restriction_{A'}\) is the function \(g: A' \to B\) such that \(f(x) = g(x)\) for all \(x \in A'\), and \(\restriction\) is called the restriction operator. 
Given a \((\VV,\EE)\)-labeled graph \(G\) with labeling function \(\ell\), and another \((\VV,\EE)\)-labeled graph \(G'\) with labeling \(\ell'\),  suppose that \(\varphi: V(G) \to V(G')\) is an isomorphism between \(G\) and \(G'\).
We say that \(\ell\) and \(\ell'\) are isomorphic if the labeling is preserved under \(\varphi\).

\subsection{Locally Checkable Labeling (LCL) Problems}
\label{sec:def-lcls}

In the distributed setting, the well-studied class of \emph{locally checkable labeling (LCL) problems} \cite{NaorS95} plays a central role, as they are those problems where the validity of a solution can be checked \emph{locally}, that is, within a constant radius \(r\). The intuition is that a node can gather its radius-\(r\) neighborhood to check whether its output satisfies the constraints of the LCL.

To formally define what it means to satisfy an LCL, we first introduce the notion of a set of constraints. This is the set of all valid labelings of a neighborhood of radius \(r\) and maximum degree \(\Delta\).
\begin{definition}[Set of constraints]
	\label{def:preliminaries:constraint}
	Let \(r, \Delta \in \mathbb{N}\) be constants.
	Consider two finite label sets \(\VV\) and \(\EE\).
	Let \(\CC\) be a finite set of pairs \((H,v_H)\), where \((H,v_H)\) is a \((\VV,\EE)\)-labeled centered graph such that the eccentricity of \(v_H\) is at most \(r\) and the degree of \(H\) is at most \(\Delta\).
	We say that \(\CC\) is an \((r,\Delta)\)-set of constraints over \((\VV,\EE)\).
\end{definition} 

After defining the notion of constraint, we can now define the notion of \emph{constraint satisfaction}. Namely, we say that a graph satisfies a set of constraints if all of its radius-\(r\) neighborhoods belong to that set.

\begin{definition}[Satisfying a set of constraints]
	\label{def:preliminaries:satisfying-constraints}
	Let \(G\) be a \((\VV,\EE)\)-labeled graph, and let \(\CC\) be an \((r,\Delta)\)-set of constraints over \((\VV,\EE)\), for some finite set of labels \(\VV,\EE\).
	The graph \(G\) satisfies \(\CC\) if the following holds:
	\begin{itemize}
		\item For every node \(u \in V(G)\), the \((\VV,\EE)\)-labeled graph \(G[\neighborhood_r[u]]\) is such that the centered graph \((G[\neighborhood_r[u]],u)\) belongs to \(\CC\).
	\end{itemize}
\end{definition}

We can now define the notion of \emph{locally checkable labeling (LCL) problems}. An LCL can be seen as a relation that specifies which pairs of input and output labelings are valid.

\begin{definition}[Locally Checkable Labeling (LCL) problems]\label{def:preliminaries:lcl-problems}
	Let \(r, \Delta \in \mathbb{N}\) be constants, and let \(\Vin\), \(\Ein\), \(\Vout\), and \(\Eout\) be finite sets of labels.
	A locally checkable labeling (LCL) problem \(\problem\) is a tuple \((\Vin, \Ein, \Vout, \Eout, \CC)\) such that the following holds:
	\begin{itemize}
		\item \(\CC\) is an \((r,\Delta)\)-set of constraints over \((\Vin \times \Vout,\Ein \times \Eout)\).  
	\end{itemize}
\end{definition}

Suppose that we are given as input a \((\Vin,\Ein)\)-labeled graph \(G\), and let \(\ell_{\inpt}\) be the input labeling function of \(G\).
Solving an LCL problem \(\problem = (\Vin, \Ein, \Vout, \Eout, \CC)\) on \(G\) means to find a labeling function \(\ell_{\oupt}\) that produces an output labeling on \(G\), turning \(G\) into a \((\Vout,\Eout)\)-labeled graph, such that the following holds:
\begin{itemize}
	\item For each node \(v \in G\), let \(\ell(v) = (\ell_{\inpt}(v), \ell_{\oupt}(v))\).
	For each half-edge \((v,e)\) in \(G\), let \(\ell((v,e)) = (\ell_{\inpt}((v,e)), \ell_{\oupt}((v,e)))\).
	Then the labeling function \(\ell\) turns \(G\) into a \({(\Vin \times \Vout,\Ein \times \Eout)}\)-labeled graph that satisfies \(\CC\) according to \cref{def:preliminaries:satisfying-constraints}.
\end{itemize}
In other words, if we are given a graph and an input labeling, solving an LCL means finding an output labeling under which the graph satisfies the set of constraints of the LCL.

\subsection{Distributed Linear Programming (LP) Problems}\label{ssec:distr-lp}

Next, we define distributed LP problems.
We consider the following distributed setting: We are given a communication graph $G=(V,E)$ and a linear program bound to $G$ of the form
\begin{align*}
	\text{optimize} \quad & \sum_{i\in\FF} c_i\cdot x_i \\
	\text{subject to} \quad & \sum_{i\in \FF}A_{j,i}\cdot x_i \unlhd b_j \quad \forall j\in \CC \\
	&  x_i \geq 0 \quad \forall i\in \FF\;,
\end{align*}
where $\FF$ is the set of variables, $\CC$ is the set of constraints, coefficients $A_{j,i},b_j,c_i$ are known locally, and the inequality $\unlhd$ can be $\leq$, $=$, or $\geq$, depending on the LP formulation. In the distributed setting, each node $v\in V$ in the network ``owns'' one or more variables $x_i\in\FF$. Furthermore, in the \local (resp. \slocal) model, each node $v\in V$ knows the local constraints and variables involving nodes within its radius-$T$ neighborhood $\NN_T[v]$. 

\subparagraph{Types of Distributed LPs.}
There are three classes of distributed LP formulations: node-based, edge-based, and node-edge-based. In the first one, each node $v\in V$ is associated with a variable $x_v$. In the second, each node has a variable $x_{(v,u)}$ for each $u\in \NN[v]$. Since an edge $(u,v)\in E$ is shared by both endpoints, we require that the two local copies coincide, i.e., $x_{(v,u)} = x_{(u,v)}$. In other words, $v$ and $u$ must agree on a common value for their shared edge variable. In the latter formulation, each node is associated with both a variable $x_v$ and a set of variables $x_{(v,u)}$ shared with its neighborhood. 

There are different types of local outputs for each class of distributed LP. For node-based LPs, the local output is a real value $x_v\in\mathbb{R}^{\geq 0}$ for each node $v\in V$, whereas for edge-based LPs, each node pair $(v,u)\in E$ agrees on and outputs a real value $x_{(u,v)}$. Consequently, for node-edge LPs, each node outputs both $x_v$ and the incident edge values $x_{(v,u)}$. These local outputs must collectively satisfy the constraint set of the LP. That is, the union of all local outputs across the network must form a globally feasible solution to the LP, meaning that the variable assignments computed and output by the nodes must jointly satisfy all constraints of the LP formulation.

We remark that, in general, distributed LPs are not LCLs. The reason is that LPs involve continuous variables and global feasibility constraints, which cannot be verified using only local information and a finite label set.

\subparagraph{Approximation Factor for Distributed LPs.}
Let $\PP$ be a linear program defined over a communication network $G = (V, E)$ with variable set $\FF$, and denote the optimal objective value by $\mathrm{OPT}$. Let $\hat{x}$ be a solution produced by a distributed algorithm after a bounded number of synchronous communication rounds. We say that the distributed algorithm achieves an \emph{$\alpha$-approximation} to $\mathcal{P}$, for some $\alpha \geq 1$, if (1) $\hat{x}$ is a feasible solution to the LP, and (2) $\sum_{i\in \FF} c_i \cdot \hat{x}_i\leq \alpha \cdot \mathrm{OPT}$ or $\mathrm{OPT}\leq \alpha\cdot \sum_{i\in \FF} c_i \cdot \hat{x}_i$ for minimization or maximization problems, respectively.
In other words, the distributed solution is within a factor $\alpha$ of the global optimum, even though each node operates with only local information.
If we are dealing with a probabilistic model of computation (e.g., \randlocal or
\nonsign; see \cref{sec:models} below), then we require that the
respective algorithm outputs an $\alpha$-approximation \emph{in expectation}.

\subparagraph{Fractional Maximum Matching.}
In this work, we consider the fractional maximum matching problem formulated as the following LP: 
\begin{align*}
	\text{maximize} \quad & \sum_{e \in E} x_e \\
	\text{subject to} \quad & \sum_{(v,u)\in E} x_{(v,u)} \leq 1 \quad \forall v\in V \\
	& 0\leq x_e \leq 1 \quad \forall e\in E
\end{align*}
Here, each variable $x_e$ corresponds to an edge $e=(v,u)\in E$ and it is ``owned'' by both endpoints $v$ and $u$ in $G$. Each node $v\in V$ is responsible for the constraint $\sum_{u\in \NN[v]} x_{(v,u)} \leq 1$, which involves all variables corresponding to the edges incident to $v$. 

Now consider any \emph{maximal matching} in the graph. By definition, a maximal matching is a matching where no additional edge can be added without violating the matching property. It is a standard result in approximation algorithms that any maximal matching is a $2$-approximation to the maximum integral matching.  Furthermore, the size of a maximum matching can be up to a factor $2/3$ smaller than the fractional maximum matching. Combining these bounds, we obtain that any maximal matching gives a feasible solution to the above LP and a solution value within a factor $3$ of the optimum.

\subsection{Models}
\label{sec:models}

In this section, we define all our computational models of interest.
\subparagraph{The \local Model.}
In the \local model of computing, we are given a distributed system of \(n\) processors (or nodes) connected through a communication network represented as a graph \(G = (V,E)\), along with an input function \(\localVar\).
Every node \(v\in V(G)\) has input data \(\localVar(v)\), which encodes the number \(n\) of nodes in the network, a unique identifier from the set \([n^c] = \{1,2,\ldots,n^c\}\),  where \(c \ge 1\) is a fixed constant, and possible inputs defined by the problem of interest (we assume nodes store both input node labels and input half-edge labels).
If computation is randomized, we call the model \emph{randomized \local} (or \randlocal), which means that \(\localVar(v)\) additionally encodes an infinite string of bits that are uniformly and independently sampled for each node, and not shared with the other nodes. If this is not the case and computation is deterministic, we call the model \emph{deterministic \local} (or \detlocal).
Computation is performed by synchronous rounds of communication. In each round, nodes can exchange messages of unbounded (but finite) size with their neighbors, and then perform an arbitrarily long (but terminating) local computation. Errors occur neither in sending messages nor during local computation.
Computation terminates when every node \(v\) outputs a label \(\outLbl(v)\).
The running time of an algorithm is the number of communication rounds, given as a function of \(n\), that are needed to output a labeling that solves the problem of interest.
In \randlocal, we also ask that the algorithm solves the problem of interest with probability at least \(1 - 1/\poly(n)\), where \(\poly(n)\) is any polynomial function in \(n\).
If an algorithm runs in \(T\) rounds and both communication and computation are unbounded, we can look at it as a function mapping radius-\(T\) neighborhoods to output labels in the deterministic case, or to a distribution of output labels in the randomized case.
Thus, we say that \(T\) is the \emph{locality} of the algorithm.

Depending on the context, we may assume that the computing units are actually the \emph{edges} of the graph, and the local variable \(\localVar(v)\) is stored inside all edges that are incident to \(v\).

\subparagraph{The \qlocal Model.}
The \qlocal model is defined like the \local model introduced above, with the following differences.
Every processor (node) can locally operate on an unbounded (but finite) number of qubits, applying any unitary transformations, and quantum measurements can be locally performed by nodes at any time. In each communication round, nodes can send an unbounded (but finite) number of qubits to their neighbors. The local output of a node still needs to be an output label encoded in classical bits. As in \randlocal, we ask that an algorithm solves a problem with probability at least \(1-1/\poly(n)\). A more formal definition of the model can be found in \cite{gavoille2009}.

\subparagraph{The \slocal Model.}
The \slocal model of computing~\cite{ghaffari2017} is a sequential counterpart of the \local model: An algorithm $\AA$ processes the nodes sequentially in an order $p=v_1,v_2,\dots,v_n$. The algorithm must work for any given order $p$. When processing a node $v$, the algorithm can query $\NN_T[v]$, and $\AA$ can read $u$'s state for all nodes $u\in \NN_T[v]$. Based on this information, node $v$ updates its own state and computes its output $y(v)$. In doing so, node $v$ can perform unbounded computation, i.e., $v$'s new state can be an arbitrary function of the queried $\NN_T[v]$. The output $y(v)$ can be remembered as a part of $v$'s state. The \emph{time complexity} $T_{\AA,p}(G,\bm{x})$ of the algorithm on graph $G$ and inputs $\bm{x}=(x(v_1),x(v_2),\dots, x(v_n))$ with respect to order $p$ is defined as the maximum $T$ over all nodes $v$ for which the algorithm queries a radius-$T$ neighborhood of $v$. The \emph{time complexity} $T_\AA$ of algorithm $\AA$ on graph $G$ and inputs $\bm{x}$ is the maximum $T_{\AA,p}(G,\bm{x})$ over all orders $p$.

\subparagraph{The \nonsign Model.}
The \nonsign model is a model of computing that abstracts from how the actual computation is happening in the network, focusing on a probabilistic description of the valid output labelings. In this model, rather than given algorithms, we are asked to produce \emph{outcomes} (also called strategies) that are functions mapping the input to a probability distribution over output labelings. In other words, given a graph and its input, a probability distribution over output labelings is assigned to the graph such that a sample from the distribution will produce a valid output labeling with high probability. The complexity of the outcome is given by its dependency radius \(T\). This means that an outcome is non-signaling beyond distance \(T\) if, for any subset \(A\) of the nodes of the graph, modifying the graph or its input at distance greater than \(T\) from \(A\) does not change the output distribution over \(A\). We proceed to give all the formal details needed to properly define the \nonsign model.

We first formally introduce the concept of \emph{outcome}. For a network \(G\), let \(\localVar\) represent the function that maps every node to its input, which includes the input labeling function, port numbers, and unique identifiers. An outcome, then, is a function mapping a network and an input \((G, \localVar)\) to a probability distribution over output labelings.
\begin{definition}[Outcome]
	Let \(\VV,\EE\) be sets of labels, 
	and let \(\FF\) be the family of all input networks \((G, \localVar)\). An outcome \(\outcome\) is a function that maps an input network \((G, \localVar)\in\FF\) to a probability distribution \(\outcome(G, \localVar)=\{(\oupt_i,p_i)\}_{i\in I}\) defined as follows:
	\begin{itemize}
		\item The set \(I\) is a set of indices.
		\item The function \(\oupt_i\) is a labeling function that maps half-edges and nodes of \(G\) to labels in \(\VV\) and \(\EE\), respectively, making \(G\) a \((\VV,\EE)\)-labeled graph.
		\item Each \(p_i\) is a non-negative probability and \(\sum_{i\in I}p_i=1\).
	\end{itemize}
\end{definition}

We say that an outcome \(\outcome\) \emph{solves} an LCL problem \(\problem\) over a family of graphs \(\FF\) with probability \(q > 0\) if, for every \(G \in \FF\) and every input data \(\localVar\), it holds that
\[
\sum_{\substack{\oupt_i \in \outcome(G, \localVar) : \\ \oupt_i\text{solves }\problem\text{ on }G }} p_i \ge q.
\]
Let \((G, \localVar)\) be an input network, and consider any subset of nodes \(S \subseteq V(G)\).
Let \(\HH(G)[S]\) be the subset of \(\HH(G)\) that contains half-edges \((v,e)\) for \(v \in S\).
The restriction of the output distribution \(\outcome(G,\localVar) = \{(\oupt_i,p_i)\}_{i \in I}\) to \(S\) is the distribution \(\outcome(G,\localVar)[S]=\{(\oupt_j,p'_j)\}_{j \in J}\), where the output-labeling functions \(\{\oupt_j\}\) assign labels only on nodes of \(S\) and on half-edges of \(\HH(G)[S]\), and the probability \(p'_j\) satisfies the following condition:
\[
p'_j=\sum_{\substack{\oupt_i \in \outcome(G, \localVar) : \\ \oupt_i \text{ coincides with } \oupt_j \\ \text{ on } S \text{ and } \HH(G)}}p_i.
\]

We now define the notion of isomorphic output distributions.
Consider two graphs \(G\) and \(G'\) such that \(\varphi:V(G)\to V(G')\) is an isomorphism.
A probability distribution \(\{(\oupt_i,p_i)\}_{i\in I}\) over output labelings  for \(G\) is isomorphic to a probability distribution  \(\{(\oupt_j,p'_j)\}_{j\in J} \) over output labelings for \(G'\) if they are preserved under the action of \(\varphi\). 

We would now like to define a special type of outcome called \emph{non-signaling outcome}. 
To this end, we first need to define the concept of \emph{view} up to distance \(T\). 
Given an input network \((G,\localVar)\) and a subset of its nodes \(A\subseteq V(G)\), consider the subgraph \(G[\neighborhood_T[A]]\) induced by \(\neighborhood_T[A]\). 
The view up to distance \(T\) of \(A\) is the pair \(\view_T(A)=(G_A,\localVar_A)\), where \(G_A\) is the graph defined as \(V(G_A)=V(G[\neighborhood_T[A]])=\neighborhood_T[A]\) and \(E(G_A)=\{(u,v) \sthat (u,v) \in \E(G[\neighborhood_T[A]]),\dist_G(u,A)<T \text{ or } \dist_G(v,A)<T\}\), and \(\localVar_A=\localVar\restriction{\neighborhood_T[A]}\). 
Intuitively, nodes in \(A\) see everything up to distance \(T\) except for the edges among the bordering nodes of \(G[\neighborhood_T[A]]\) (but they can see the labels of the half-edges incident to them).
In general, for two arbitrary graphs \(G\) and \(H\) and subsets of nodes \(A \subseteq V(G), B \subseteq V(H)\), we say that a function \(\varphi : V(G) \to V(H)\) is an isomorphism between \(\view_T(A)=(G_A,\localVar_A)\) and  \(\view_T(B)=(G_B,\localVar_B)\) if \(\varphi\) is an isomorphism between \(V(G_A)\) and \(V(G_B)\) and \(\localVar_A=\localVar_B\circ\varphi\). 
Now, a non-signaling outcome is defined as follows.
\begin{definition}[Non-signaling outcome]
	\label{def:ns-outcome}
	Let \(\outcome\) be an outcome, \(G\) and \(H\) be graphs, \(\varphi : V(G) \to V(H)\) be a function, and \(T\in\mathbb{N}\). Outcome \(\outcome\) is \emph{non-signaling beyond distance} \(T\) if, for any two subsets of nodes \(A_G \subseteq V(G), A_H \subseteq V(H)\) such that \(\varphi\) is an isomorphism between \(\view_T(A_G)\) and \(\view_T(A_H)\), the restricted distributions \(\outcome(G, \localVar_G)[A_G]\) and \(\outcome(H, \localVar_H)[A_H]\) are isomorphic under \(\varphi\).
\end{definition}
Alternatively, we can also say that \(\outcome\) has locality \(T\).
Running \(T\)-round classical or \qlocal algorithms, both with or without shared resources, yields output-labeling distributions that are non-signaling outcomes with locality \(T\). 

The \emph{\nonsign model} is, thus, a computational model where the input is a network with input \((G,\localVar)\), and an LCL problem \(\problem\) is solved if there exists a non-signaling outcome \(\outcome\) that solves \(\problem\) with success probability at least \(1 - 1/\poly(n)\), where \(n=|V(G)|\).

When the input of a problem is clear from the context, we will omit the input network \((G,\localVar)\), writing \(\outcome(G)\).
Observe that all the concepts of views and of restrictions of outcomes can be naturally defined via half-edges instead of nodes, especially when dealing with problems that only ask us to label half-edges (such definitions will be used later in \cref{sec:lower-bound}).

We further assume that all probability distributions and output labelings that define a non-signaling \(\outcome\) are \emph{computable}.
This is because a proper quantum-LOCAL algorithm can be implemented by a quantum circuit, which can be simulated in a classical computer (with costly computation).
Hence, its output distribution is computable, and we can restrict ourselves to computable output-labeling distributions.


\section{Dequantization for Distributed Linear Programming Problems}\label{sec:dequant}

In this section, we prove our first result, i.e., that distributed \nonsign
(and in particular also \qlocal) has no advantage over \detlocal for
distributed linear programming problems.
\begin{theorem}
	\label{thm:ns-lp}
	Let $\lpproblem$ be a distributed linear programming problem that admits a non-signaling distribution over $\alpha$-approximations with locality~$T$.
	Then there exists a deterministic \local algorithm that finds an
	$\alpha$-approximation of $\lpproblem$ with locality~$T$.
\end{theorem}
For simplicity, we will consider only the case where $\lpproblem$ is a node-based
problem. It is clear how to extend the proof to the other classes of distributed LPs. The idea of the proof is relatively simple:
We first note that for any distribution of $\alpha$-approximations of a linear program~$\lpproblem$, the \emph{expectation} is also an $\alpha$-approximation; this follows directly from the convexity of a linear program and the linearity of expectation. A detailed proof can be found in Appendix~\ref{apx:missing-proofs}.
\begin{lemma}
	\label{lem:expectation-approximation}
	Let $\lpproblem$ be a linear program, and let $\outcome$ be a distribution over $\alpha$-approximations of $\lpproblem$.
	Then $\hat x = \expect{\outcome}$ is also an $\alpha$-approximation of $\lpproblem$.
\end{lemma}

Then we show that there exists a \local algorithm that can locally compute this expectation, given access to the distribution. Intuitively, the algorithm gathers its radius-$T$ neighborhood, where $T$ is the locality of the \nonsign distribution, then calls the outcome function $\outcome$ on its neighborhood to obtain a distribution of output labelings, and finally outputs the expected value of such a distribution. Again, the technical details are deferred to Appendix~\ref{apx:missing-proofs}.
\begin{lemma}
	\label{lem:local-expectation}
	Let $\outcome$ be a computable \nonsign distribution with locality~$T$ over graph family~$\FF$.
	Then there exists a \local algorithm with locality~$T$ that computes the expected outcome of this distribution everywhere.
\end{lemma}

We are now ready to state the proof of \cref{thm:ns-lp}.

\begin{proof}[Proof of \cref{thm:ns-lp}]
	Let~$\lpproblem$ be a distributed linear programming problem and let~$\outcome$ be computable \nonsign distribution with locality~$T$ over $\alpha$-approximations of~$\lpproblem$.
	By \cref{lem:local-expectation}, we have a \local algorithm~$\mathcal{A}$ that computes the expectation of~$\outcome$ locally everywhere.
	As the outcome is a vector over local elements, its expectation is a vector over the expectations of the local elements.
	Hence, $\mathcal{A}$ computes the expectation of~$\outcome$.
	By \cref{lem:expectation-approximation}, this is an $\alpha$-approximation of~$\lpproblem$.
\end{proof}


\section{Separation Between SLOCAL and \nonsign}\label{sec:separation}

In this section, we prove that there exists an LCL problem $\Pi$ that has complexity $O(\log n)$ in the SLOCAL model and complexity $\omega(\log n)$ in the \nonsign model.
\begin{theorem}\label{thm:separation}
	There exists an LCL problem $\Pi$ that has complexity $O(\log n)$ in the deterministic SLOCAL model and $\Omega\left(\log n \cdot \sqrt{\frac{\log n}{\log \log n}}\right)$ in the \nonsign model.
\end{theorem}
We devote the rest of this section to proving \Cref{thm:separation}.

\subsection{Overview}
In order to define the problem $\Pi$, we borrow ideas from \cite{balliu2025quantum-lcl}. We start by giving a recap of the main ideas presented in \cite{balliu2025quantum-lcl}.

\subparagraph{Recap of Previous Results.}
The authors of \cite{balliu2025quantum-lcl} introduced the notion of \emph{linearizable problems}, which are locally checkable problems that are not necessarily LCLs. They proved that, if there exists some linearizable problem $P$ with some complexity $f(n)$ for some function $f$, then there exists some LCL problem $\Pi = \lift(P)$ with some complexity $f'(n)$, where $f'$ depends on $f$. For small-enough $f$, the function $f'$ is a multiplicative factor $\Theta(\log n)$ larger than $f$. Interestingly, both the quantum complexity and the standard complexity are increased by this $\Theta(\log n)$ factor.
In more detail, the authors of \cite{balliu2025quantum-lcl} proved the following:
\begin{enumerate}
	\item In \cite{balliu2024quantum}, it has been shown that there exists a problem $P$ with quantum complexity $O(1)$ and randomized LOCAL complexity $\Omega(\min\{\Delta,\log_\Delta \log n\})$. By taking a suitable value of $\Delta$, this result implies a lower bound of $\Omega(\frac{\log \log n}{\log \log \log n})$.
	\item The problem $P$ can be expressed as a linearizable problem.
	\item The authors defined a function $\lift$ that takes as input a linearizable problem $P$ and returns an LCL problem $\Pi = \lift(P)$.
	\item The authors showed that $\Pi  = \lift(P)$ has the following complexities:
	\begin{itemize}
		\item $O(\log n)$ in quantum-LOCAL.
		\item $\Omega(\log n \cdot \frac{\log \log n}{\log \log \log n})$ in \randlocal.
	\end{itemize}
\end{enumerate}
Note that, while the problem $P$ itself does not have any super-constant lower bound when $\Delta = O(1)$, this construction allows us to nevertheless obtain a problem $\Pi = \lift(P)$ with a super-constant lower bound as a function of $n$ on graphs in which $\Delta = O(1)$.

\subparagraph{Our Approach.}
In essence, we show that the construction of \cite{balliu2025quantum-lcl} also preserves complexities in SLOCAL and in \nonsign, and we will use the maximal matching problem, phrased as a linearizable problem, as $P$. We will obtain the following:
\begin{enumerate}
	\item Our problem $P$ will be the maximal matching problem, phrased as a linearizable one.
	\item The problem $P$ has complexity $O(1)$ in SLOCAL, even on graphs of unbounded degree.
	\item The problem $\Pi = \lift(P)$ has complexity $O(\log n)$ in SLOCAL.
	\item Maximal matching has a \randlocal lower bound of $\Omega(\sqrt{\frac{\log n}{\log \log n}})$, proved as part of the KMW lower bound \cite{KMW}.
	\item Maximal matching is a $3$-approximation of fractional maximum matching, and the KMW bound, which is more general, holds for this latter problem as well. Thus, by \Cref{thm:ns-lp}, the lower bound for maximal matching also holds in \nonsign.
	\item We prove that, in \nonsign, an upper bound of $o(\log n \cdot \sqrt{\frac{\log n}{\log \log n}})$ for $\Pi$ would imply an upper bound of $o(\sqrt{\frac{\log n}{\log \log n}})$ for $P$, contradicting the KMW lower bound.
\end{enumerate}

In order to achieve this, we prove a result similar to the one of \cite{balliu2025quantum-lcl}, with the difference that we consider SLOCAL upper bounds, and \nonsign lower bounds.

\subsection{The Definition of \texorpdfstring{\boldmath $\Pi = \lift(P)$}{Pi = lift(P)}}\label{ssec:lift}
We summarize the definition of $\Pi = \lift(P)$ that appeared in \cite{balliu2025quantum-lcl}. 
At a high level, the problem $\Pi$ is an LCL with inputs (i.e., nodes and node-edge pairs have input labels that come from a finite set) that is defined as a combination of two LCL problems:
\begin{itemize}
	\item The LCL problem $\Pi^{\lbadgraph}$, which is a problem defined on any graph.
	\item The LCL problem $\Pi^{\lpromise}$, which is a problem defined on some specific class $\mathcal{G}$ of graphs labeled with some input (which is the input of $\Pi$). Note that, in order for a graph $G$ to be in $\mathcal{G}$, the input given to the nodes of $G$ must satisfy some specific local constraints.
\end{itemize}
In particular, $\Pi^{\lbadgraph}$ asks us to produce some output that satisfies, among others, the following properties:
\begin{itemize}
	\item Each node is either marked (labeled with some specific output labels) or unmarked (labeled $\bot$).
	\item If $G \in \mathcal{G}$, then no node of $G$ is marked.
\end{itemize}
Moreover, it is shown that there exists a deterministic $O(\log n)$ LOCAL algorithm $\mathcal{A}$ solving $\Pi^{\lbadgraph}$ on any graph $G$, such that the output of $\mathcal{A}$ satisfies that each connected component induced by unmarked nodes is in $\mathcal{G}$.
Specifically, the authors of \cite{balliu2025quantum-lcl} proved the following.
\begin{lemma}[\cite{balliu2025quantum-lcl}]\label{lem:pibadgraph-valid}
	Let $G \in \mathcal{G}$. Then, the only valid solution for $\Pi^{\lbadgraph}$ on $G$ is the one assigning $\bot$ to all nodes.
\end{lemma}
\begin{lemma}[\cite{balliu2025quantum-lcl}]\label{lem:pibadgraph-ub}
	Let $G$ be any graph. There exists a solution for $\Pi^{\lbadgraph}$ where each connected component induced by nodes outputting $\bot$ is a graph in $\mathcal{G}$. Moreover, such a solution can be computed in $O(\log n)$ deterministic rounds in the LOCAL model.
\end{lemma}

The problem $\Pi$ is defined such that it is first required to solve $\Pi^{\lbadgraph}$, and then, on each connected component induced by unmarked nodes, it is required to solve $\Pi^{\lpromise}$. We will later describe the problem $\Pi^{\lpromise}$, the definition of which will depend on the given linearizable problem $P$.  Now we argue that, in order to prove our lower and upper bounds, we can restrict our attention to graphs that are in $\mathcal{G}$ and to the problem $\Pi^{\lpromise}$.

In \cite{balliu2025quantum-lcl}, the quantum-LOCAL upper bound for $\Pi$ is obtained as follows:
\begin{itemize}
	\item First, apply \Cref{lem:pibadgraph-ub}. That is, in $O(\log n)$ deterministic LOCAL rounds, we obtain a solution for $\Pi^{\lbadgraph}$ satisfying that each connected component induced by nodes outputting $\bot$ is a graph in $\mathcal{G}$.
	\item Then, use a quantum algorithm to solve $\Pi^{\lpromise}$ on each connected component induced by nodes outputting $\bot$. This requires $O(\log n)$ quantum rounds. 
\end{itemize}
Since an $O(\log n)$ deterministic LOCAL algorithm can be directly executed in SLOCAL (i.e., SLOCAL is at least as strong as LOCAL), and since an SLOCAL algorithm obtained by composing two different SLOCAL algorithms has an asymptotic complexity equal to the sum of the complexities of the two composed algorithms, it is clear from the quantum algorithm that, in order to provide an $O(\log n)$ SLOCAL algorithm for $\Pi$, it is sufficient to provide an $O(\log n)$ SLOCAL algorithm for $\Pi^{\lpromise}$ on graphs that are in $\mathcal{G}$.

The randomized LOCAL lower bound for $\Pi$ is obtained by considering graphs $G \in \mathcal{G}$. On these graphs, by \Cref{lem:pibadgraph-valid}, the only valid solution for $\Pi^{\lbadgraph}$ is the one assigning $\bot$ to all nodes. By the definition of $\Pi$, this implies that, on $G$, it is required to solve $\Pi^{\lpromise}$. For our \nonsign lower bound, we will follow the exact same strategy.

\subsection{The Graph Family \texorpdfstring{\boldmath $\mathcal{G}$}{G}}
In order to define the family $\mathcal{G}$, we need to first introduce the notion of \emph{proper instances}.
At a high level, a proper instance is a graph that can be obtained by starting from some graph $G'$ (which is not necessarily a simple graph) and replacing nodes with some \emph{gadgets} according to some rules. Then, a graph $G \in \mathcal{G}$ will be obtained by labeling a proper instance in some specific way.
In the following, we report the definition of some objects as given in \cite{balliu2025quantum-lcl}.
The basic building block is the notion of \emph{tree-like gadget}, of which we give an example in \Cref{fig:tree-like-and-octopus} (top), while the formal definition can be found in Appendix~\ref{apx:missing-definitions}~(\cref{def:tree-like-gadget}).

\begin{figure}[t]
	\centering
	\captionsetup[subfigure]{justification=centering}
	\begin{subfigure}{.48\textwidth}
		\centering
		\includegraphics[page=1,width=0.6\textwidth]{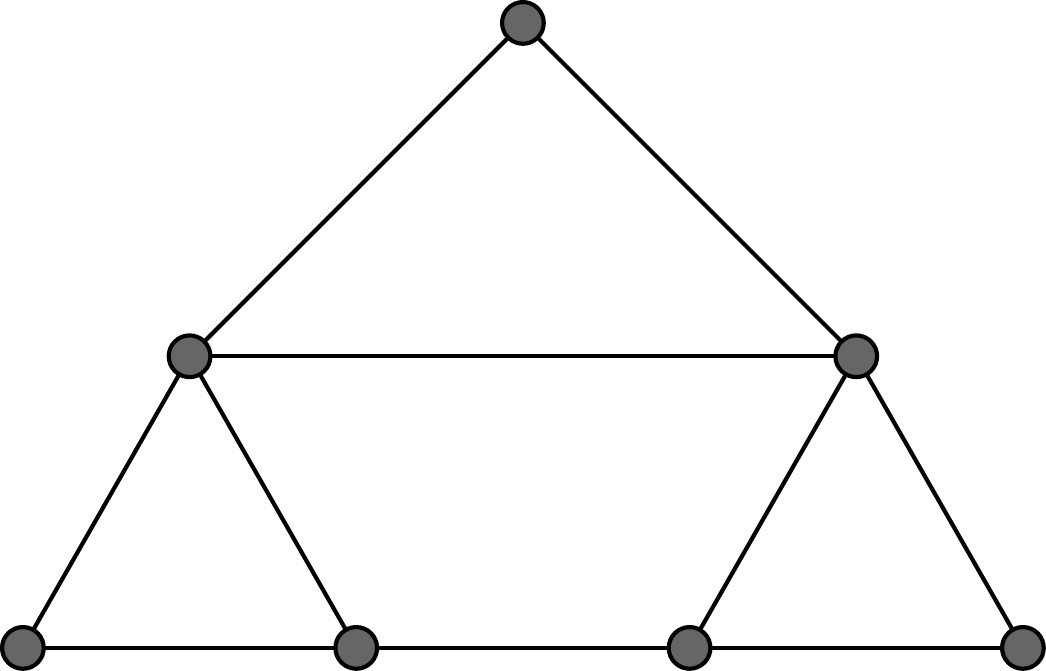}\\[4pt]
		\includegraphics[page=2,width=0.6\textwidth]{figs-gadgets.pdf}
		\caption{}
		\label{fig:tree-like-and-octopus}
	\end{subfigure}%
	\begin{subfigure}{.48\textwidth}
		\centering
		\includegraphics[page=3,width=\textwidth]{figs-gadgets.pdf}
		\caption{}
		\label{fig:proper}
	\end{subfigure}
	\caption{\textbf{(a)} A tree-like gadget at the top, and an octopus gadget at the bottom. \textbf{(b)} A proper instance at the bottom and, at the top, the graph obtained by contracting each octopus gadget into a single node. }
	\label{fig:tree-like-and-octopus-and-proper}
\end{figure}

The next building block is called \emph{octopus gadget}. At a high level, an octopus gadget is obtained by starting from a tree-like gadget, and connecting one or two additional tree-like gadgets to each ``leaf'' of the tree-like gadget. See \Cref{fig:tree-like-and-octopus} (bottom) for an example and Appendix~\ref{apx:missing-definitions} (\cref{def:octopus-gadget}) for a formal definition.

We can now define the family of proper instances. An example is shown in \Cref{fig:proper}.
\begin{definition}[Proper instance \cite{balliu2025quantum-lcl}]\label{def:proper-instance}
	Let $G=(V,E)$ be a graph.
	We say that \(G\) is a \emph{proper instance} if there exists a node labeling function \(\lambda: V \to \{\textsf{intra-octopus}, \textsf{inter-octopus}\}\) with the following properties.
	\begin{enumerate}
		\item Every connected component in the subgraph induced by nodes labeled \(\textsf{intra-octopus}\) is an octopus gadget (according to \cref{def:octopus-gadget}).
		\item The subgraph induced by nodes labeled  \(\textsf{inter-octopus}\) does not contain any edge.
		\item A node $v$ labeled \(\textsf{intra-octopus}\) is connected to a node labeled \(\textsf{inter-octopus}\) if and only if $v$ has coordinates \((w-1,0)\) in the port gadget $P$ containing $v$, where $w$ is the height of $P$ (that is, $v$ is the left-most leaf of the port gadget containing $v$).
	\end{enumerate}
\end{definition}

The authors of \cite{balliu2025quantum-lcl} proved that proper instances are \emph{locally checkable}, in the sense that there exist a set of local constraints $\mathcal{C}$ and finite sets of labels $\VV$ and $\EE$ for which an arbitrary graph $G$ can be \((\VV,\EE)\)-labeled such that the constraints $\CC$ to be satisfied on all nodes, if and only if $G$ is a proper instance. More precisely, they proved the following statements.
\begin{lemma}[\cite{balliu2025quantum-lcl}]
	Let \(G\) be any non-empty connected graph that is \((\VV, \EE)\)-labeled such that \(\CC\) is satisfied at all nodes. 
	Then \(G\) is a proper instance according to \cref{def:proper-instance}.
\end{lemma}
\begin{lemma}[\cite{balliu2025quantum-lcl}]\label{lem:proper-can-be-labeled}
	Let \(G\) be a proper instance as defined in \cref{def:proper-instance}.
	Then there exists a \((\VV, \EE)\)-labeling of \(G\) that satisfies the constraints in \(\CC\) at all nodes.
\end{lemma}

We are now ready to define the family $\mathcal{G}$.
\begin{definition}
	A \((\VV, \EE)\)-labeled graph $G$ is in $\mathcal{G}$ if and only if the constraints in \(\CC\) are satisfied at all nodes.
\end{definition}

\subsection{Linearizable Problems}
In order to define $\Pi^{\lpromise}$, we first need to introduce the notion of \emph{linearizable problems}.
A linearizable problem $\Pi^{\linearizable} = (\Sigma, (F,L,P),B)$ is defined as follows.
Let us consider a hypergraph described as its bipartite incidence graph, where white nodes represent the original nodes, black nodes represent hyperedges, and we are also given an ordering of the edges of the white nodes. Intuitively, a problem $\Pi^{\linearizable} = (\Sigma, (F,L,P),B)$ is linearizable if it is possible to encode it as another LCL where the constraints for white nodes can be expressed solely in terms of consecutive edges in the ordering. More specifically, sets $F$ and $L$ define which labels are allowed for the first and last edges, respectively, while set $P$ contains pairs of labels that can appear consecutively. In other words, $P$ specifies the combinations in which labels can be assigned to a node and its successor in the ordering. We refer to Appendix~\ref{apx:missing-definitions} (\cref{def:Pi_linearizable}) for a formal definition.

We will use \emph{maximal matching} as the running example of a linearizable problem. This is an important step to reach our goal, as we will use the fact that maximal matching can be expressed as a linearizable problem to separate \slocal from \nonsign. Maximal matching is a problem defined on graphs, and hence, we will describe a linearizable problem on hypergraphs of rank $2$. In this case, the black constraint describes the edge constraints, and the white constraint describes the node constraints. An example of a solution to the maximal matching problem encoded as a linearizable problem is provided in \Cref{fig:matching-linearizable}, and the following lemma formally states the existence of a linearized version of maximal matching. We defer the details of the proof to Appendix~\ref{apx:missing-proofs}.

\begin{lemma}\label{lem:matching-linearizable}
	There exists a linearizable problem $\Pi^{\linearizable} = (\Sigma, (F,L,P),B)$ satisfying the following:
	\begin{itemize}
		\item A solution for $\Pi^{\linearizable}$ can be converted into a maximal matching in $0$ deterministic LOCAL rounds.
		\item A solution for maximal matching can be converted into a solution for $\Pi^{\linearizable}$ in $0$ deterministic LOCAL rounds.
	\end{itemize} 
\end{lemma}
\begin{figure}[t]
	\centering
	\includegraphics[page=4,width=0.6\textwidth]{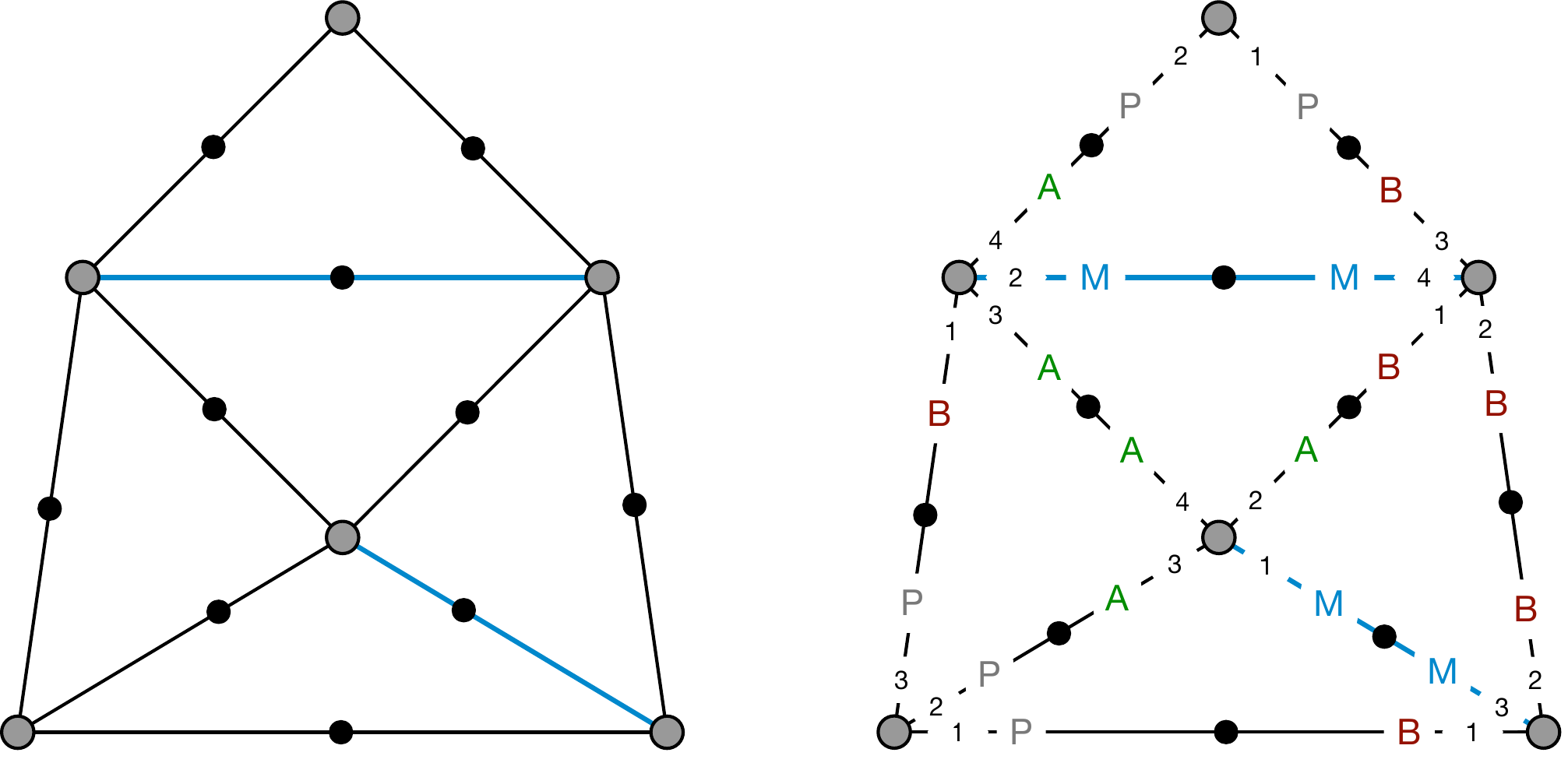}
	\caption{On the left, we show a solution to the maximal matching problem, where black nodes represent hyperedges of rank 2 and blue edges are in the matching. On the right, the same solution is encoded as a solution to $\Pi^{\linearizable}$.}
	\label{fig:matching-linearizable}
\end{figure}

We now connect the notion of a linearizable problem with the problem  $\Pi^{\lpromise}$ that we aim to define. The problem $\Pi^{\lpromise}$ is defined in \cite{balliu2025quantum-lcl} as an LCL problem (i.e., by describing its local constraints). This problem is defined as a function of a given problem $\Pi^{\linearizable} = (\Sigma, (F,L,P),B)$. For our purposes, we do not need the details of the definition of $\Pi^{\lpromise}$, and it is sufficient to state the properties that any valid solution needs to satisfy.
Informally, $\Pi^{\lpromise}$ is defined such that, if we contract each octopus gadget into a single white node, and we treat each inter-cluster node as a black node, it must hold that, in the resulting graph, we get a solution for $\Pi^{\linearizable}$. The formal definition of $\Pi^{\lpromise}$ can be found in Appendix~\ref{apx:missing-definitions} (\cref{def:p-promise}) or in \cite{balliu2025quantum-lcl}.

\subsection{SLOCAL Upper Bound and \nonsign Lower Bound}\label{sec:lower-bound}

We now establish both an upper bound in the SLOCAL model and a lower bound in the \nonsign model, deferring the formal proofs to Appendix~\ref{apx:missing-proofs}. For \slocal, we have the following upper bound.
\begin{lemma}\label{lem:slocal-ub}
	Let $T(n)$ be an upper bound on the SLOCAL complexity of $\problem^{\linearizable}$ that holds also if the given graph contains parallel edges. Then the SLOCAL complexity of $\Pi$ is upper bounded by $O(T(n) \log n)$.
\end{lemma}

For \nonsign, we have the following lower bound.
\begin{lemma}\label{lemma:ns:lower}
	Let $T(n)$ be a lower bound on the locality of $\Pi^{\linearizable}$ in \nonsign with failure probability \(p(n)\), which is a non-increasing function of \(n\) bounded above by some constant $q < 1$. Then any non-signaling outcome for $\Pi^{\lpromise}$ with failure probability at most \(p(n)\) requires locality $\Omega(T(n^{1/3}) \log n)$.
\end{lemma}

The following lemma is the key ingredient for the proof of \cref{lemma:ns:lower}.
\begin{lemma}\label{lemma:ns:lift}
	Let \(\Pi^{\linearizable}\) be any linearizable problem, and consider the LCL problem \(\Pi = \lift(\Pi^{\linearizable})\).
	Suppose that there exists an outcome \(\outcome\) that is non-signaling beyond distance \(T(n)\) that solves \(\Pi\) with failure probability \(p(n)\), which is a non-increasing function of \(n\) bounded above by some constant $q < 1$.
	Then we can construct an outcome \(\outcome'\) that solves \(\Pi^{\linearizable}\) with failure probability at most \(p(n)\) and is non-signaling beyond distance \(T'(n) = O(T(n^3) / \log n)\).
\end{lemma}

The proof of \Cref{lemma:ns:lower} is now straightforward.

\begin{proof}[Proof of \cref{lemma:ns:lower}]
	If any non-signaling outcome for \(\Pi\) with failure probability \(p'(n) \le p(n)\) has locality \(o(T(n^{1/3}) \cdot \log n)\), then we can construct a non-signaling outcome that solves \(\Pi^{\linearizable}\) with failure probability at most \(p'(n)\) and has locality \(o(T(n))\) by \cref{lemma:ns:lift}, which is a contradiction with the fact that we assumed $T(n)$ to be a lower bound on the locality of $\Pi^{\linearizable}$.
\end{proof}

\subsection{Instantiating Our Construction}
We are now ready to prove \Cref{thm:separation}, referring to Appendix~\ref{apx:kmw} for a high-level description of the KMW lower bound and to \cite{coupette2021breezing} for a detailed exposition.
We consider the maximal matching problem, and by \Cref{lem:matching-linearizable}, there exists a linearizable problem $P$ that is equivalent to maximal matching.
In SLOCAL, maximal matching, and hence $P$, can be solved in one deterministic round by a trivial greedy algorithm. Hence, by \Cref{lem:slocal-ub}, the deterministic SLOCAL complexity of $\Pi = \lift(P)$ is $O(\log n)$.

For a lower bound, recall that KMW gives us a lower bound of $\Omega(\sqrt{\log n / \log\log n})$ for $O(\log \Delta)$-approximation of fractional maximum matching (see \Cref{thm:kmw}) on graphs of degree $\Delta = 2^{\Theta(\sqrt{\log n \log \log n})}$.
Combining \Cref{thm:ns-lp} with the KMW lower bound and the fact that fractional maximum matching can be expressed as a linear program, we obtain the following corollary.
\begin{corollary}
	There does not exist a non-signaling distribution for $O(\log \Delta)$-approximation of fractional maximum matching with locality $o(\sqrt{\log n / \log\log n})$  on graphs of degree $\Delta = 2^{\Theta(\sqrt{\log n \log \log n})}$.
\end{corollary}

Since a maximal matching is a $2$-approximation of a maximum matching, and a maximum matching is a $\frac{3}{2}$-approximation of a maximum fractional matching, we obtain the following.
\begin{corollary}\label{cor:mm-ns}
	There does not exist a non-signaling distribution for maximal matching with locality $o(\sqrt{\log n / \log\log n})$.
\end{corollary}

By combining \Cref{cor:mm-ns} with \Cref{lemma:ns:lower}, we obtain that in \nonsign, $\Pi$ has complexity $\Omega(\log n \cdot \sqrt{\frac{\log n}{\log \log n}})$, as desired.

\section*{Acknowledgements} This work was supported in part by the MUR (Italy) Department of Excellence 2023 - 2027 for GSSI, the European Union - NextGenerationEU under the Italian MUR National Innovation Ecosystem grant ECS00000041 - VITALITY – CUP: D13C21000430001, the PNRR MIUR research project GAMING “Graph Algorithms and MinINg for Green agents” (PE0000013, CUP D13C24000430001), and the Research Council of Finland, Grants 363558 and 359104.


\printbibliography


\newpage
\appendix
\begin{center}
	{\LARGE \textbf{Appendix}}
\end{center}
\section{Omitted Definitions}
\label{apx:missing-definitions}

\begin{definition}[Tree-like gadget \cite{balliu20lcl-randomness,balliu2025quantum-lcl}]\label{def:tree-like-gadget}
	A graph $G$ is a \emph{tree-like} gadget of height $\ell$ if it is possible to assign coordinates $(l_u, k_u)$ to each node $u\in G$, where
	\begin{itemize}
		\item $0\le l_u < \ell$ denotes the depth of $u$ in the tree, and
		\item $0\le k_u < 2^{l_u}$ denotes the position of $u$ (according to some order) in layer $l_u$,
	\end{itemize}
	such that there is an edge connecting two nodes $u,v\in G$ with coordinates $(l_u, k_u)$ and $(l_v, k_v)$ if and only if:
	\begin{itemize}
		\item $l_u = l_v$ and $|k_u - k_v | = 1$, or
		\item $l_v = l_u-1$ and $k_v = \lfloor \frac{k_u}{2} \rfloor$, or
		\item $l_u = l_v-1$ and $k_u = \lfloor \frac{k_v}{2} \rfloor$.
	\end{itemize}
\end{definition}
\begin{definition}[Octopus gadget \!\cite{balliu2025quantum-lcl}]\label{def:octopus-gadget}
	Let $x \ge 1$ be a natural number, and \(\eta = (\eta_0, \dots,\allowbreak \eta_{2^{x-1}-1}) \) a vector of \(2^{x-1}\) entries in \(\{1,2\}\).
	Let \(W = \{w_{(i,j)}\}_{(i,j) \in I}\) be a family of positive integer weights, where $I$ is the set containing all pairs $(i,j)$ satisfying \({(i,j) \in \{0,1,\dots,2^{x-1}-1\}}\) \(\times\{1,2\}\)  and \(j \le \eta_{i}\).
	
	A graph \(G = (V,E)\) is an \((x,\eta,W)\)-\emph{octopus gadget} if there exists a labeling \(\lambda: V \to \LL = I \cup \{\text{root}\}\) of the nodes of \(G\) such that the following holds.
	\begin{enumerate}
		\item For each element $y \in \LL$, let \(G_y\) be the subgraph of \(G\) induced by nodes labeled with \(y\). Then, for all \(y \in \LL\), \(G_y\) must be a tree-like gadget according to \cref{def:tree-like-gadget}.
		\item For all \(y,z \in \LL\) such that \(y \neq z\), \(G_y\) and \(G_z\) must be disjoint.
		\item \(G_{\text{root}}\) has height \(x\) and, for all \((i,j) \in I\), \(G_{(i,j)}\) has height \(w_{(i,j)} \in W\).
		\item For all \((i,j) \in I\), there is an edge connecting the node of \(G_{(i,j)}\) that has coordinates \((0,0)\) with the node of \(G_{\text{root}}\) that has coordinates \((x-1,i)\).
	\end{enumerate}
	\(G_{\text{root}}\) is called the \emph{head-gadget} and, for all \((i,j) \in I\), \(G_{(i,j)}\) is called a \emph{port gadget}.
\end{definition}

\begin{definition}[Linearizable problem \cite{balliu2025quantum-lcl}]\label{def:Pi_linearizable}
	Let $H$ be a hypergraph, and let $G$ be its bipartite incidence graph. Let the nodes of $G$ corresponding to the nodes of $H$ be called \emph{white} nodes, and let the nodes of $G$ corresponding to the hyperedges of $H$ be called \emph{black} nodes. 
	\begin{itemize}
		\item The task requires to label each edge of $G$ with a label from some finite set $\Sigma$.
		\item There is a list of allowed \emph{black node configurations} $B$, which is a list of multisets of labels from $\Sigma$ that describes valid labelings of edges incident on a black node. We say that a black node \emph{satisfies the black constraint} if the multiset of labels assigned to its incident edges is in $B$. It is assumed that the rank of $H$, and hence the maximum degree of black nodes, is a constant. 
		\item Constraints on white nodes are described as a triple $(F,L,P)$, where $F$ (which stands for \emph{first}) and $L$ (which stands for \emph{last}) are finite sets of labels, and $P$ (which stands for \emph{pairs}) is a finite set of ordered pairs of labels. In this formalism, it is assumed that an ordering on the incident edges of a white node is given, and it is required that:
		\begin{itemize}
			\item The first edge is labeled with a label from $F$;
			\item The last edge is labeled with a label from $L$;
			\item Each pair of consecutive edges must be labeled with a pair of labels from $P$.
		\end{itemize}
		We say that a white node \emph{satisfies the white constraint} if its incident edges are labeled in a valid way.
	\end{itemize}
	When solving a linearizable problem in the distributed setting, it is assumed that each node knows whether it is white or black.
\end{definition}

\begin{definition}[The definition of $\Pi^{\lpromise}$ of \cite{balliu2025quantum-lcl}, rephrased]
	\label{def:p-promise}
	Given a graph $G \in \mathcal{G}$, the problem $\Pi^{\lpromise}$ requires to label the nodes of the graph as follows:
	\begin{itemize}
		\item Each node that is not in a port gadget must be labeled $\bot$.
		\item Each node that is in a port gadget must be labeled with a label from $\Sigma$.
		\item Nodes that belong to the same port gadget must be assigned the same label.
		\item Let $g$ be an octopus gadget, and let $\ell_1,\ldots,\ell_d$ be the labels assigned to the $d$ port gadgets of $g$, according to the natural left-to-right order of the port gadgets of $g$. It must hold that $\ell_1 \in F$, $\ell_d \in L$, and $(\ell_i,\ell_{i+1}) \in P$ for all $i$.
		\item Let $v$ be an inter-cluster node, and let $\ell_1, \ldots, \ell_r$ be the labels assigned to the nodes of its $r$ incident port gadgets. It must hold that $\{\ell_1,\ldots, \ell_r\} \in B$.
	\end{itemize}
\end{definition}

\section{Omitted Proofs}
\label{apx:missing-proofs}

\begin{proof}[Proof of \cref{lem:expectation-approximation}]
	To show that $\hat x$ is an approximation of $\lpproblem$, we need to establish (1)~that $\hat x$ is feasible and (2)~that $\hat x$ gives the correct approximation ratio. 
	To see the feasibility of $\hat x$, observe that the feasibility constraints are of the form
	\begin{equation*}
		\sum_{i\in \FF}A_{j,i}\cdot x_i \unlhd b_j \quad \forall j\in \CC \;.
	\end{equation*}
	Plugging in $\hat x$ and fixing $j \in \CC$ gives us
	\begin{equation*}
		\sum_{i\in \FF}A_{j,i}\cdot \hat x_i
		= \sum_{i\in \FF}A_{j,i}\cdot \expect{O_i}
		= \mathbb{E}\Bigl[\underbrace{\sum_{i\in \FF}A_{j,i}\cdot O_i}_{ \unlhd b_j}\Bigr]
		\unlhd b_j \;.
	\end{equation*}
	The first equality holds by definition, the second equality holds by linearity of expectation, and the conclusion holds by the fact that $\outcome$ is a distribution over feasible solutions and the monotonicity of expectation.
	Hence, $\hat x$ is a feasible solution for $\lpproblem$.
	
	It is left to show that $\hat x$ is also an $\alpha$-approximation.
	Again, we can plug $\hat x$ into the target function, obtaining
	\begin{equation*}
		\sum_{i\in\FF} c_i\cdot \hat x_i
		= \sum_{i\in\FF} c_i\cdot \expect{O_i}
		= \mathbb{E}\Bigl[\sum_{i\in\FF} c_i\cdot O_i\Bigr] .
	\end{equation*}
	As each outcome of $\outcome$ is an $\alpha$-approximation, we can invoke the linearity and monotonicity of expectation and get that this new target is also an $\alpha$-approximation of~$\lpproblem$.
\end{proof}

\begin{proof}[Proof of \cref{lem:local-expectation}]
	We give the description for the \local algorithm~$\mathcal{A}$:
	Node~$v$ gathers its radius-$T$ neighborhood~$\neigh_T[v]$; this can be done with locality~$T$.
	It then constructs an arbitrary graph~$G' \in \FF$ such that the neighborhood of node~$v' \in V(G')$ is isomorphic to~$\neigh_T[v]$. 
	Note that such a graph always exists, as the graph the algorithm is being run on is one such graph. 
	Now~$v$ invokes the distribution~$\outcome$ on graph~$G'$ to compute the distribution of outputs for node~$v'$ and, in particular, the outcome.
	Node~$v$ then outputs this expected outcome and~halts.
	
	It remains to argue that this algorithm computes the expected outcome of~$\outcome$ everywhere.
	This follows directly from the definition of the \nonsign distributions as the marginals of nodes~$v$ and~$v'$ coincide, and hence their expectations must also coincide.
	Moreover, the choice of graph~$G'$ does not affect this marginal distribution.
\end{proof}

\begin{proof}[Proof of \Cref{lem:matching-linearizable}]
	We define the problem as follows. The label set $\Sigma$ is defined as $\Sigma = \{\lM, \lB,\lA,\lP\}$, where $\lM$, $\lB$, $\lA$, and $\lP$ stand for \emph{matched}, \emph{before}, \emph{after}, and \emph{pointer}, respectively.
	The list of allowed black configurations is defined as
	
	$B = \{\{\lM,\lM\}, \{\lP,\lB\}, \{\lP,\lA\}, \{\lB,\lB\}, \{\lB,\lA\}, \{\lA,\lA\}\}$.
	Then, $F$ is defined as $F = \{\lM,\lB,\lP\}$, $L$ is defined as $L = \{\lM,\lA,\lP\}$, and $P$ is defined as $P = \{(\lB,\lB),(\lB,\lM),(\lM,\lA),(\lA,\lA),(\lP,\lP)\}$.
	
	We observe that the defined problem  $\Pi^{\linearizable}$  satisfies the following properties.
	\begin{itemize}
		\item In any valid solution, for each node $v$, if we treat the labels assigned to the half-edges incident to $v$ as a string (according to the ordering assigned to the half-edges incident to $v$), we obtain that such a string must satisfy the regular expression $\lP^* \mid \lB^* \lM \lA^*$. We call nodes satisfying $\lP^*$ \emph{unmatched} nodes and nodes satisfying $\lB^* \lM \lA^*$ \emph{matched} nodes. Moreover, we call an edge \emph{matched} if both its half-edges are labeled $\lM$.
		\item Since $\{\lP,\lP\} \notin B$, we get that, in any valid solution, unmatched nodes cannot be neighbors.
		\item Since the label $\lM$ appears only in the pair $\{\lM,\lM\}$, and since each matched node must have exactly one incident matched edge, matched edges form an independent set.
	\end{itemize}
	This implies that a solution for $\Pi^{\linearizable}$  can be converted into a maximal matching in $0$ deterministic LOCAL rounds.
	
	We now show that a maximal matching can be converted into a solution for $\Pi^{\linearizable}$ in $0$ deterministic LOCAL rounds. A solution for $\Pi^{\linearizable}$ can be computed as follows:
	\begin{itemize}
		\item Each unmatched node labels $\lP$ all its incident half-edges.
		\item Each matched node labels $\lM$ its incident half-edge $e$ that is part of the matching, $\lB$ all edges that come before $e$ in the given ordering, and $\lA$ all edges that come after $e$ in the given ordering.
	\end{itemize}
	It is easy to see that the computed solution satisfies the constraints of $\Pi^{\linearizable}$.
\end{proof}

\begin{proof}[Proof of \cref{lem:slocal-ub}]
	Let $\mathcal{A}$ be an SLOCAL algorithm for $\problem^{\linearizable}$ with complexity $T(n)$.
	We show how to use $\mathcal{A}$ to solve $\Pi^{\lpromise}$ with SLOCAL complexity $O(T(n) \log n)$. As argued in \Cref{ssec:lift}, this implies a solution for $\Pi$ with the same asymptotic SLOCAL complexity.
	
	Let $G \in \mathcal{G}$ be the graph in which we want to solve  $\Pi^{\lpromise}$.
	Consider the virtual bipartite graph $\widehat{G}$ obtained by contracting each octopus gadget into a single node (see \Cref{fig:proper}), that~is:
	\begin{itemize}
		\item For each octopus gadget $g$ of $G$, there is a white node $v_g$ in $\widehat{G}$.
		\item For each inter-cluster node $b$ of $G$, there is a black node $u_b$ in $\widehat{G}$.
		\item For each edge connecting an inter-cluster node $b$ of $G$ to an octopus gadget $g$ of $G$, there is an edge between $v_g$ and $u_b$ in $\widehat{G}$.
	\end{itemize}
	Note that $\widehat{G}$ may contain parallel edges. 
	Since the diameter of a valid octopus gadget is clearly upper bounded by $O(\log n)$, we get that the distances in $G$ are at most an $O(\log n)$ factor larger than distances in $\widehat{G}$. Thus, it is possible to simulate the execution of an SLOCAL algorithm for $\widehat{G}$ with an $O(\log n)$ multiplicative overhead on $G$. We use $\mathcal{A}$ to solve $\problem^{\linearizable}$ on $\widehat{G}$. For each octopus gadget $g$, we assign the solution of the $i$-th port of $g_v$ to the nodes of the $i$-th port gadget of $g$, according to the natural left-to-right order of the port gadgets of $g$.	The output clearly satisfies the constraints of $\Pi^{\lpromise}$, and the runtime is upper bounded by $O(T(n) \log n)$.
\end{proof}

\begin{proof}[Proof of \Cref{lemma:ns:lift}]
	By hypothesis, there exists a non-signaling outcome \(\outcome\) that solves \(\Pi\) with failure probability \(p(n)\).
	Consider any input hypergraph \(F\) for \(\Pi^{\linearizable}\) of size \(n\), and let $G$ be the bipartite incidence graph of $F$.
	We construct a graph $G'$ as a function of $G$ as follows.
	\begin{itemize}
		\item For each white node $v$ of degree $d$ of $G$, we put an octopus gadget $g_v$ with $d$ port gadgets, each of height $\Theta(\log n)$, into $G'$.
		\item For each black node $u$ of $G$, we put an inter-octopus node $b_u$ in $G'$.
		\item Let $v$ be an arbitrary white node in $G$, and let $\{v,u\}$ be its $i$-th incident edge, according to the given ordering. We put, in $G'$, an edge connecting $b_u$ to the left-most leaf of the $i$-th port gadget of $g_v$, according to the natural left-to-right order of the port gadgets of~$g_v$.
	\end{itemize}
	By construction, $G'$ is a proper instance, and by \Cref{lem:proper-can-be-labeled} it can be labeled such that $G' \in \mathcal{G}$. In the following, we assume that $G'$ is labeled in such a way. Hence, \(G'\) is an input instance for \(\Pi^\lpromise\).
	Moreover, we get that if \(G\) has \(n\) nodes, then \(G'\) has \(n \le N \le n^3\) nodes. 
	
	We define a function $f_G$ as follows. Let $v$ be a white node of $G$, and let $r$ be the root node of the $i$-th port gadget of $g_v$. The function $f_G$ maps $r$ to the $i$-th edge incident to~$v$, according to the given ordering. It is straightforward to see that $f_G$ maps solutions for \(\Pi^\lpromise\) on \(G'\) to solutions for \(\Pi^\linearizable\) on \(G\).
	
	Let \(V_r(G')\) be the domain of \(f_G\), and let \(\Sigma_\lpromise\) be the set of output labels for \(\Pi^\lpromise\) for the nodes in \(V_r(G')\).
	Let \(\{(\oupt_i, p_i)\}_{i \in I}\) be the output distribution that \(\outcome\) defines on \(G'\) for~\(\Pi\). As discussed in \Cref{ssec:lift}, this is also a valid output distribution for $\Pi^{\lpromise}$.
	Note that \(G\) is a bipartite incidence graph and \(\Pi^\linearizable\) asks only to label edges of \(G\).
	This means that only half-edges of \(F\) are labeled.
	We now define an outcome \(\outcome'\) on \(G\) just by describing output labelings on edges of \(G\) (which correspond to half-edges of \(F\))---more specifically, we set \(\outcome'(G) = \{(\oupt_i  \circ f^{-1}_G, p_{i}) \}_{i \in I}\).
	
	First, it is clear that the sum of all \(p_{i}\) in \(\outcome'(G)\) is exactly 1.
	Furthermore, it is straightforward to check that \(\outcome'\) has failure probability at most \(p(n)>0\).
	If not, by construction of \(\outcome'\), then \(\outcome(G')\) has failure probability strictly greater than \(p(n)\) for \(\Pi^\lpromise\), which is a contradiction because \(\outcome\) has failure probability \(p(N) \le p(n)\) by monotonicity of \(p\).
	
	We now claim that \(\outcome'\) is non-signaling beyond distance \(T'(n)\).
	Recall that each octopus gadget in \(G'\) represents a white node \(v\) of \(G\), and neighboring octopus gadgets represent neighboring white nodes of \(G\).
	Hence, for every subset of edges \(A\) of \(E(G)\), the distribution \(\outcome'(G)[A]\) is defined only by \(\{(\oupt_i, p_i)\}_{i \in I}[f^{-1}_G(A)]\).
	Let \(k(n)\) be the height of a port gadget, and observe that  \(k(n) = \Theta(\log n)\).
	Suppose we modify \(G\) outside the radius-\(T'(n)\) view of \(A\) and obtain a graph \(H\) with a subset of edges \(A_H\) such that \(\view_0(A)\) is isomorphic to \(\view_0(A_H)\) and \(\view_{T'(n)}(A)\) is isomorphic to \(\view_{T'(n)}(A_H)\).
	Notice that, as before, \(H\) also defines a proper instance \(H' \in \mathcal{G}\) for \(\Pi^\lpromise\).
	However, because of the isomorphic regions between \(G\) and \(H\), we get that \(\view_0(f^{-1}_G(A))\) is isomorphic to \(\view_0(f^{-1}_H(A_H))\), and \(\view_{T'(n)\cdot k(n)}(f^{-1}_G(A))\) is isomorphic to \(\view_{T'(n)\cdot k(n)}(f^{-1}_H(A_H))\), since each port gadget has height at least \(k(n)\).
	We impose that \(T'(n) \cdot k(n) \ge T(N)\), which is equivalent to asking that \(T'(n) \ge T(N) / k(n)\).
	Since the distribution \(\outcome\) is non-signaling beyond distance \(T(N)\), we have that \(\outcome(G')[\view_0(f^{-1}_G(A))]\) is the same distribution as \(\outcome(H')[\view_0(f^{-1}_H(A))]\).
	Hence, \(\outcome'(G)[A]\) and \(\outcome'(H)[A]\) are equal, and \(\outcome'\) is non-signaling beyond distance \(T'(n)\).
	Note that it is sufficient to take \(T'(n) = O(T(n^3) / \log n) \), since \(T(N)\) is non-decreasing in \(N\) and \(n \le N \le n^3\).
\end{proof}


\section{KMW in a Nutshell}\label{apx:kmw}

To instantiate our construction, we use a lower bound for approximating maximum matchings based on the KMW bound \cite{KMW}. For completeness, we now describe the high-level idea of this bound, state it formally, and sketch the key components of its construction that are relevant for our work. 
See Coupette and Lenzen~\cite{coupette2021breezing} for a detailed exposition and a simplified proof.

The KMW bound establishes that there exist graphs with $n$ nodes and maximum degree $\Delta = 2^{\Theta(\sqrt{\log n \log \log n})}$ 
on which  $\Omega(\sqrt{\log n/\log \log n})$ (expected) communication rounds are required 
to obtain polylogarithmic approximations to a minimum vertex cover, minimum dominating set, or maximum matching. 
\begin{theorem}[\cite{KMW}]\label{thm:kmw}
	There does not exist a randomized \local algorithm providing an $O(\log \Delta)$-approximation of fractional maximum matching with locality $o(\sqrt{\log n / \log\log n})$ on graphs with degree $\Delta = 2^{\Theta(\sqrt{\log n \log \log n})}$.
\end{theorem}

The KMW bound holds under both randomization and approximation, 
and it extends to symmetry-breaking tasks like finding maximal independent sets or maximal matchings via straightforward reductions. 

At the core of the bound lies a class of high-girth graphs
constructed from a blueprint, the \emph{Cluster Tree}, 
which arranges differently-sized independent sets of nodes as a tree and prescribes that node sets adjacent in the tree are connected via biregular bipartite graphs. 
Both blueprints and graphs are parametrized by the number of communication rounds~$k$, 
and they are designed to enable an \emph{indistinguishability argument}: 
For a given~$k$, the associated Cluster Tree graph contains two independent sets of nodes, one large and one small, such that both sets of nodes have isomorphic radius-$k$ neighborhoods, 
but only the small set of nodes is needed to solve a given \emph{covering} problem.
This forces any algorithm to select a large fraction of the large node set into the solution (in expectation), 
yielding a poor approximation ratio. 

The construction extends to \emph{packing} problems by taking two copies of a Cluster Tree and additionally prescribing that each node in the first copy is connected to its counterpart in the second copy. 
Importantly, the graphs arising from Cluster Trees are bipartite by design. 
In bipartite graphs, the optimal fractional solution and the optimal integral solution coincide for both minimum vertex cover and maximum matching, 
and by K\H{o}nig's theorem \cite{konig1916graphok}, the solution sets to both problems have the same cardinality. 
Moreover, the two-copy construction of Cluster Trees for packing problems has a natural bicoloring such that the large cluster in the first copy and the small cluster in the second copy have the same color---%
i.e., providing the bicoloring keeps the indistinguishability argument intact. 
Hence, the KMW bound is inherently a bound for (bipartite) fractional problems.

\end{document}